\documentclass{article}
\usepackage[english]{babel}
\usepackage[cp1250]{inputenc}
\usepackage[T1]{fontenc}
\usepackage{lmodern}
\usepackage{amsmath}
\usepackage{amssymb}
\usepackage{epsfig}
\usepackage{latexsym}
\usepackage{amsfonts}
\usepackage{times}
\usepackage{color}
\usepackage{graphicx}
\usepackage{float}
\usepackage{placeins}
\usepackage{algorithm}
\usepackage[noend]{algpseudocode}
\usepackage{algorithmicx}
\usepackage{amsmath}
\usepackage{tabularx}
\usepackage{curves}
\usepackage{pgf}
\usepackage{verbatimbox}
\usepackage{fancyvrb}
\usepackage{amsthm}
\usepackage{MnSymbol}
\usepackage{amssymb} 
\usepackage{relsize}
\usepackage{url}

\usepackage{graphicx}
\usepackage{caption}
\usepackage{subcaption}
\newtheorem{trd}{Theorem}

\theoremstyle{definition}
\newtheorem{defn}{Definition}

\floatstyle{ruled}
\newfloat{program}{!h}{lop}
\floatname{program}{Algorithm}

\newcommand{\RR}{\mathbb{R}}
\newcommand{\NN}{\mathbb{N}}
\newcommand{\NNN}{\overline \NN}
\newcommand{\RRR}{\overline {\RR_0^+}}
\newcommand{\ZZZ}{\overline {\mathbb Z}}
\newcommand{\RRRp}{\overline {\mathbb R}}

\newcommand{\network}[1]{\mathcal{#1}}
\newcommand{\vertices}[1]{\mathcal{#1}}
\newcommand{\edges}[1]{\mathcal{#1}}

\newcommand{\functions}[1]{\mathcal{#1}}
\newcommand{\Time}{\mathcal{T}}

\newcommand{\cmdkey}{\raisebox{-.025em}{\includegraphics[height=.7em]{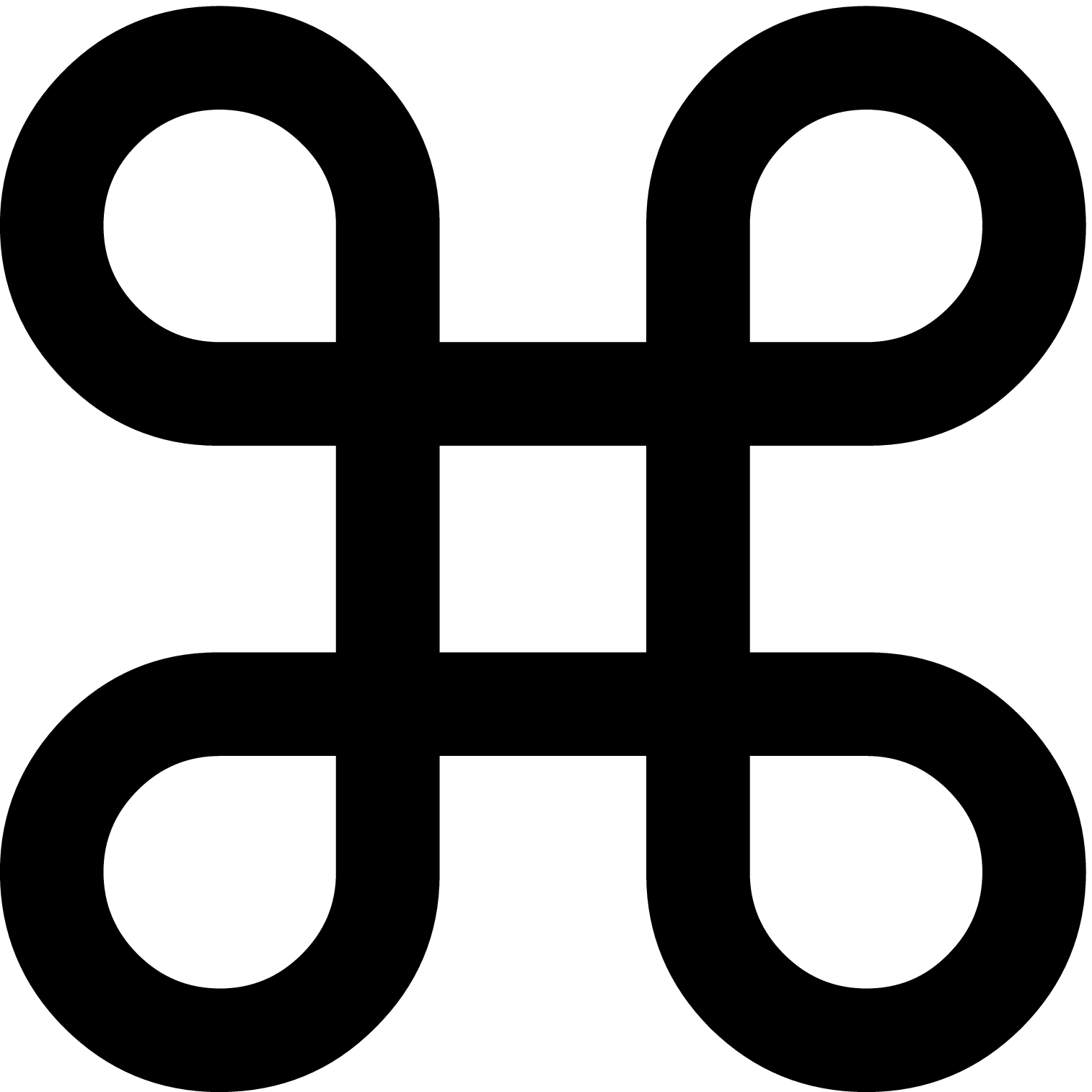}}}
\newcommand{\Mw}{\mathop{\raisebox{-1.5pt}{\mbox{$\Box$\kern-.55em\raisebox{2.5pt}{{\tiny $r$}}\kern2.9pt}}}}
\newcommand{\Mv}{\mathop{\raisebox{-1.5pt}{\mbox{$\Box$\kern-.55em\raisebox{2.5pt}{{\tiny $h$}}\kern2.9pt}}}}

\newcommand{\splus}{\mathop{\raisebox{-1.5pt}{\mbox{$\Box$\kern-.6em\raisebox{2.5pt}{{\tiny $+$}}\kern2.9pt}}}}
\newcommand{\sdot}{\mathop{\raisebox{-1.5pt}{\mbox{$\Box$\kern-.6em\raisebox{2.5pt}{{\tiny $\times$}}\kern2.9pt}}}}

\def\RR{{\mathbb R}}

    \makeatletter
    \def\Ddots{\mathinner{\mkern1mu\raise\p@
    \vbox{\kern7\p@\hbox{.}}\mkern2mu
    \raise4\p@\hbox{.}\mkern2mu\raise7\p@\hbox{.}\mkern1mu}}
    \makeatother

\begin{document}

	\title{Semirings for temporal network analysis}
	

	\author{Selena Praprotnik \footnote{Faculty of Electrical Engineering, University of Ljubljana, Tržaška cesta 25, 1000 Ljubljana, Slovenia} \footnote{Corresponding author: selena.praprotnik@fe.uni-lj.si},
Vladimir Batagelj\footnote{IMFM, Jadranska 19, 1000 Ljubljana, Slovenia and 
			University of Primorska, IAM, Muzejski trg 2, 6000 Koper, Slovenia}}
	
	\maketitle
	
	\begin{abstract}
		{In the article, we describe a new algebraic approach to the temporal network analysis based on the notion of temporal quantities. We define  the semiring for computing the foremost journey and the traveling semirings for the analysis of temporal networks where the latency is given, the waiting times are arbitrary, and some other information on the links are known. We use the operations in the traveling semiring to compute a generalized temporal betweenness centrality of the nodes that corresponds to the importance of the nodes with respect to the ubiquitous foremost journeys in a temporal network.}
		\\
		2000 Math Subject Classification: 05C25, 
		68R10, 
		90B10, 
		91D30, 
		16Y60. 
	\end{abstract}

\section{Introduction}

Network analysis is used for different purposes in operations research, social sciences and many other scientific fields. A lot of research is done in communication networks, logistics,  and the internet. The interest in network analysis increased in recent times, mostly due to the availability of big data and the global interest in data analysis. The growth of the internet and the amount of information available gave rise to many methods for the analysis of big data and sparse networks. In the last decade, there has been growing interest in temporal networks -- the networks which change over time.

In a temporal network, the presence and the activity of nodes and links can change through time. Temporal data was added to networks in different scientific fields, for example transport systems \cite{Transp,Wardrop} and project management (CPM, Pert) in operations research \cite{Moder}.  An overview of temporal network analysis is given in \cite{TNsur, TNbook,HolmeRev}.

A lot of research is still confused with the terminology and the terms used in communication network analysis, transport networks, computer networks, etc.~that are similar or even the same, define the same phenomenon with different notation and different words. For example, temporal distance \cite{Xuan_JoFoCS2002}, reachability time \cite{Holme_PRE2005}, latency of the information, and other terms name the same thing in different areas. The same thing happens with journeys \cite{Xuan_JoFoCS2002} that are named temporal paths, time respecting paths or paths with schedules by other authors.

There is no established formal description of temporal networks. The common point of all current research is the time component and that the changes of the network are one of the key information about the network.

The beginnings of temporal network analysis are based on time slices of the network  \cite{slice}. The temporal network is represented as a sequence of static networks, representing the state of the temporal network at a chosen time point (interval).

Two different approaches aim to unify temporal networks theory in a way that could be used for all the different uses. One is the time-aggregated graph from \cite{series/sbcs/GeorgeK13}. The other is the time-varying graph from \cite{TVGsur}.

We feel that both descriptions lack the possibility of adding arbitrary information to the network nodes or links. They are both describing the presence with explicit functions which also seems too complex. In \cite{TQSON} and \cite{PraprotnikAMC}, we proposed a new way for the temporal network description which remedies both of these shortcomings.

In the article, we shortly explain our description of temporal networks and study the case of temporal networks that is an extension of static networks and of temporal networks with zero latency and zero waiting times described in our previous articles \cite{TQSON,PraprotnikAMC}. We define a mathematical model for the description of temporal networks that allows for the presence / activity of the nodes / links and for the node properties and the link weights to change through time. The amount of the information that can be described with our representation of temporal networks is not limited. We  construct semirings with operations that allow us to define and compute a simple node centrality measure in a temporal network.

Most of the static network analysis based on paths has been difficult to generalize to the case of temporal networks because of the obvious differences -- in static networks the shortest path always includes the shortest subpaths which is not true in temporal networks (we address this issue in more detail at the end of the article). Also, these measures cannot be generalized with the time slices approach as the temporal network can be disconnected at every time point and connected through time (think of the network of e-mail messages). The analysis of path based indices has to be done on dynamic networks that include the latency information.

For some special cases, there were steps taken to compute shortest, fastest, and foremost journeys \cite{Xuan_JoFoCS2002}. But the complexity of the standard problems of network analyisis can be a lot greater in temporal networks. For example, the problem of strongly connected components in temporal networks is NP complete \cite{SCC, journals/corr/abs-1106-2134}.

With this article, we make a step towards unifying temporal networks description and to adding information to the nodes and links of the temporal network. We also provide a way to combine different information in a useful manner. One such example is the generalization of the betweenness centrality.

In Section \ref{sec:definition}, we present some basic definitions and notation used in the rest of the paper.

In Section \ref{sec:semiring}, we define semirings and describe their use in network analysis. We give some examples that we need for the description and better understanding of temporal semirings.

In Section \ref{sec:temporalSemi}, we present the definitions of our new approach to the temporal network analysis. We introduce the notion of temporal quantities and the temporal semirings for the analysis of temporal networks with zero latency and zero waiting time. We introduce the semiring of increasing functions and explain how it is used in computing the foremost journeys -- we get the first arrival semiring. The traveling semirings take into account additional network information, besides the latency.

We describe the application of these semirings on some generic temporal networks. We used the Python library TQ that we started writing in our previous articles \cite{PraprotnikAMC,TQSON} and we extended it to include the operations in the first arrival semiring and in the traveling combinatorial semiring. We are also developing a user friendly program called Ianus for an easy access to the library options. The program and the library are freely available at\ \url{http://vladowiki.fmf.uni-lj.si/doku.php?id=tq}.

In Section \ref{sec:between}, we explain a possible use of the traveling semiring -- two generalizations of the betweenness centrality. We extended the library TQ so that it can be used to compute the first arrival betweenness and the first arrival betweenness with exclusion in any network described in Ianus format. We test the proposed centrality on a part of the bus schedule network of Ljubljana, Slovenia.

We conclude with directions for future work in Section \ref{sec:conclusion}. Our work opens a lot of different future research possibilities.

\section{Definitions and notation} \label{sec:definition}

\begin{defn} A \emph{graph} $\mathcal G$ is an ordered pair of sets $(\vertices{V},\edges{L}),$ the set $\mathcal V$ is the set of \emph{nodes} and the set  $\mathcal L$ is the set of \emph{links} between nodes. The links between the nodes $u$ and $v$ can be \emph{directed (arcs)} $(u,v)$ or \emph{undirected (edges)} $\{u,v\}.$ With $\ell(u,v)$ we tell that the link $\ell$ goes from $u$ to $v.$ If for an arc $\ell$ it holds $\ell(u,v)$  we say that $\ell$ starts at $u$ and ends at $v.$ 
\end{defn}

With $n$ we denote the number of nodes $|\mathcal V|$ and with $m$ the number of links $|\mathcal L|$. We assume that  $n$ and $m$ are finite.  

\begin{defn} 
	A \emph{network} $\mathcal N = (\mathcal V, \mathcal L, \mathcal P, \mathcal W)$ consists of the graph $\mathcal G = (\mathcal V, \mathcal L)$ with additional information about the values (\emph{weights}) of links $\mathcal W$ and  the values (\emph{properties}) of the nodes $\mathcal P.$
\end{defn}

\begin{defn}
	A \emph{walk} in a graph $\mathcal G$ with a \emph{start} at the node $v_0$ and an \emph{end } at the node $v_p$ is a finite alternating sequence of nodes and links $$\pi = v_0\ell_1v_1\ell_2v_2\dots \ell_pv_p$$ iff $\ell_i (v_{i-1},v_i), \; i = 1,2,\dots,p.$ The \emph{length} of a walk is the number $p$ of links it contains. The sequence $\pi$ is a \emph{semiwalk} iff the direction of the links is not important, that is $\ell_i(v_{i-1},v_i)$ or $\ell_i(v_i,v_{i-1})$ for all $i = 1,2,\dots,p.$ A walk is \emph{closed} iff it starts and ends at the same node, $v_0 = v_p.$ A walk without repeating nodes is  an \emph{elementary walk} or a \emph{path}.
\end{defn}

\begin{defn} A  \emph{value matrix} $\mathbf{A}$ of a network $\mathcal N=(\mathcal V,\mathcal L,w)$ is defined as
	$$\mathbf{A}=[a_{uv}]_{u,v\in \mathcal V} = \left\{\begin{array}{l l}
	w(u,v), & (u,v) \in \mathcal L,\\
	0, &  \mbox{otherwise},
	\end{array}\right.$$
	where $w(u,v)$ denotes the weight associated with the link $(u,v).$
\end{defn}

In our notation $0 \in \NN.$ We denote $\NNN = \NN \cup \{\infty \},$ $\ZZZ = \mathbb Z \cup \{\pm \infty \},$ $\RRRp = \RR\cup \{\pm\infty \}$  and $\RRR = \RR_0^+\cup \{\infty \}.$

\section{Semirings} \label{sec:semiring}

Semirings are frequently used in network analysis \cite{doi:10.2200/S00245ED1V01Y201001CNT003, GaN, Dolan:2013:FSF:2544174.2500613, Gondran:2008:GDS:1386688, Mohri:2002:SFA:639508.639512, zimmerman}. In this section, we describe  semirings that are used most frequently and are later generalized for the analysis of temporal networks.

\begin{defn} Let $a,b,c \in A.$ The set $A$ with binary operations addition $\oplus$ and multiplication $\odot,$ neutral element $0$ and unit $1,$ denoted with $A(\oplus, \odot, 0, 1),$  is a \emph{semiring}, when the following conditions hold:
\begin{itemize}
	\item the set $A$ is a commutative monoid for the addition $\oplus$ with a neutral element 0 (the addition is commutative, associative and  $a \oplus 0  = a$ for all  $a \in A$);
	\item the set $A$ is a monoid for the multiplication $\odot$ with the unit $1$ (the multiplication is associative and $a \odot 1 = 1 \odot a = a$ for all  $a \in A$);
	\item the addition distributes over the multiplication
	\begin{align*}
	a \odot (b \oplus c) &= (a \odot b) \oplus (a \odot c) \quad \mbox{ and } \quad
	(a \oplus b) \odot c = (a \odot c) \oplus (b \odot c);
	\end{align*}
	\item the element 0 is an absorbing element or zero for the multiplication
	$$a \odot 0 = 0 \odot a = 0 \mbox{ for all } a \in A.$$
\end{itemize}
\end{defn}

In all cases, we assume precedence of the multiplication over the addition. The last point in the definition of semirings is omitted by some authors. We need it in order to construct a matrix semiring over the semiring $A.$ If all the points in the definition, except for the last one, hold for a given set $A,$ it can be extended with the element $\cmdkey,$ for which by definition
\[ a \oplus \cmdkey = \cmdkey \oplus a = a \quad \mbox{and} \quad
a \odot \cmdkey = \cmdkey \odot a = \cmdkey  \]
holds for all $a \in A \cup \{\cmdkey\}.$ In the extended set $A_{\scriptsize\cmdkey} = A \cup \{\cmdkey\}$  the element $\cmdkey$ is a zero by the definition and $(A_{\scriptsize\cmdkey},\oplus,\odot,\cmdkey,1)$ is a semiring.

\begin{defn}
	A semiring is \emph{complete} iff the addition is well defined for countable sets and the distributivity laws still hold.
	\end{defn}
	
\begin{defn}
	The addition is \emph{idempotent} iff $a\oplus a=a$ for all $a \in A$.
\end{defn}

\begin{defn}
	A complete semiring $(A,\oplus, \odot, 0,1)$ is \emph{closed} iff an additional unary operation \emph{closure} $\star$ is defined in it and
	$$a ^\star = 1 \oplus (a \odot a^\star) = 1 \oplus (a^\star \odot a) \mbox{ for all } a \in A.$$  
	We define a \emph{strict closure} $\overline{a}$ in a closed semiring as
	$$\overline{a} = a \odot a^\star.$$
\end{defn}

There can be different closures in the same semiring. A complete semiring is closed when the closure is defined with
	\begin{equation} \label{eq:zaprtje}
	a^\star = \bigoplus_{k \geq 0} a^k.
	\end{equation}
In the rest of the article the term closure describes the operation from the equation \eqref{eq:zaprtje}.

\begin{defn}
A  semiring $(A,\oplus, \odot, 0,1)$ is \emph{absorptive}  iff for every
 $a,b,c \in A$ it holds
$$(a \odot b) \oplus (a \odot c \odot b) = a \odot b.$$
\end{defn}

Because of the distributivity and the existence of the unit, it is enough to check that $1 \oplus c = 1$ for every $c \in A$  for the validity of the absorption law.  In absorptive semirings also $a^\star = 1$ for all $a \in A.$ An absorptive semiring  is idempotent.

\begin{defn}
Over the semiring $(A,\oplus,\odot,0,1),$ we construct the \emph{semiring of square matrices} $A^{n\times n}$ of order $n$ which consist of the elements from $A.$  The addition and the multiplication in the matrix semiring are defined in the usual way:
\begin{align*}
(\mathbf{A} \oplus \mathbf{B})_{ij} &= a_{ij} \oplus b_{ij} \quad \mbox{ and } \quad
(\mathbf{A} \odot \mathbf{B})_{ij} = \bigoplus_{k=1}^{n}a_{ik}\odot b_{kj}, \quad i,j = 1,2, \dots, n.
\end{align*}
\end{defn}

Note that the operations on the left hand side operate in the matrix semiring $A^{n \times n}$ and the operations on the right hand side operate in the underlying semiring $A.$

For computing the closure $\mathbf{A}^\star$ of the network value matrix $\mathbf{A}$ over a complete semiring $(A,\oplus,\odot,0,1),$ the Fletcher's algorithm can be used. It is described in \cite{Fletcher:1980:MGA:358876.358884}.

\subsection{The use of semirings in network analysis}

In network analysis, semirings are used to combine weights on the links of the network. Combining the weights, we can observe different network properties. There are two basic cases -- combining weights of two parallel links between two nodes or the weights of two sequential links between three nodes. The weights on the parallel links are combined using the semiring addition and the weights on the sequential links are combined using the semiring multiplication. A graphical representation is given in Figure \ref{fig:semiG}. Using the semiring operations, the weights of links can be extended to walks and to sets of walks in the network \cite{Batagelj94semiringsfor}.

\begin{figure}[!ht] 
	\begin{center}
\includegraphics{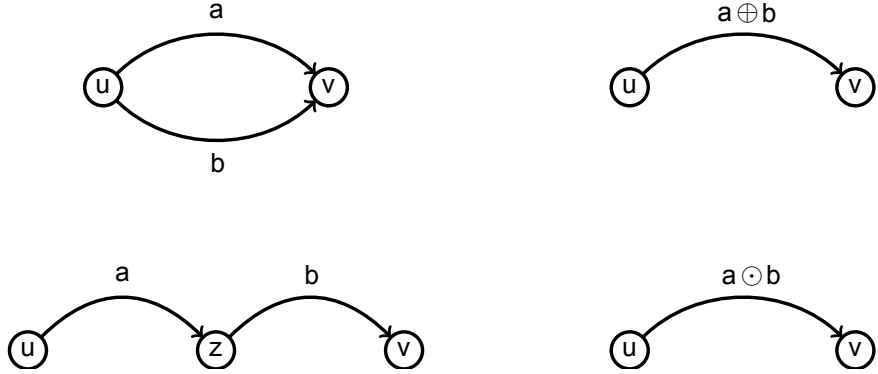}
	\end{center}
	\caption{The semiring addition and the semiring multiplication in networks.}
	\label{fig:semiG}
\end{figure}

\subsubsection{Combinatorial semiring}

The combinatorial semiring is the semiring of the natural numbers for the usual addition and multiplication $(\NNN,+,\cdot,0,1).$ In some cases other number sets are used, for example  $\RRR.$ This semiring is complete and closed for $a^\star = \sum_{k\geq 0} a^k.$ It is not absorptive and the addition is not idempotent.

In network analysis, the combinatorial semiring is used when the weights of links represent the number of ways to traverse them.
The semiring addition and multiplication correspond to the rule of sum and the rule of product used in combinatorics \cite{comb}.

\subsubsection{Shortest paths semiring}

The shortest paths semiring is defined as $(\RRR, \min,+,\infty,0).$ It is complete, commutative (i.e.~also the semiring multiplication is commutative), and absorptive. It is closed and  $a^\star = \min\{0,a+a^\star\} = 0$ for all $a \in \RRR.$
If the set $\NNN$ is used instead of $\RRR$, the semiring is called \emph{tropical}.

The shortest paths semiring is used in the classical shortest paths problem:

A network $\mathcal N=(\mathcal V,\mathcal L,w)$ with weights on links $w:\mathcal L\to \RRR$ and a (source) node $s\in \mathcal V$ are given. The value $w(u,v)$ represents the length of the link from  $u$ to $v.$ We would like to compute all  lengths of the shortest paths from $s$ to other nodes  $v\in \mathcal V\setminus \{s\}.$ The usual solution is using dynamic programming: Define $d(s) = 0$ and compute the distances to other nodes $v\in \mathcal V \setminus \{s\}$ using Bellman's equation
\begin{equation} \label{eq:BellmanFord}
d(v) = \min_{u \in \mathcal V} \{d(u)+w(u,v)\}.
\end{equation} 

\subsubsection{Geodetic semiring} \label{sec:geodez}

In a set $A = \RRR \times \NNN$ the addition
$$(a,i)\oplus (b,j) = \left(\min (a,b), \left\{\begin{array}{l l}
i,  & a < b\\
i+j, & a = b \\
j, & a>b
\end{array}\right. \right)$$
and the multiplication
$$ (a,i)\odot (b,j) = (a+b,i \cdot j)$$
are defined. For these operations $\big(A,\oplus,\odot,(\infty,0),(0,1)\big)$ is a complete closed semiring \cite{Batagelj94semiringsfor} for the closure
$$(a,i)^\star = \left\{\begin{array}{l l}
(0,\infty), & a = 0, i \neq 0,\\
(0,1), &  \mbox{otherwise}.
\end{array}\right.$$
It is called a geodetic semiring. It is not idempotent.

The geodetic semiring is a combination of the shortest paths semiring and the combinatorial semiring. It is used to compute the length and the number of the shortest paths between pairs of nodes.

\section{Semirings for temporal networks} \label{sec:temporalSemi}

\begin{defn} A \emph{temporal network} $\network{N} =(\vertices{V},\edges{L}, \Time, \functions{P},\functions{W})$
is an ordinary (static) network $(\vertices{V},\edges{L},\functions{P},\functions{W})$ with an added time dimension $\Time$. Similarly to static networks, $\mathcal G ( \mathcal V, \mathcal L,\Time)$ is a graph which can now change through time. In temporal networks, $\mathcal W$ are temporal weights on the links, and $\mathcal P$ are temporal properties of the nodes. The set $\Time$ of time points $t \in \Time$ is a lifetime of the network. The lifetime $\Time$ is usually a subset of integers  $\Time \subseteq \ZZZ$ or a subset of reals $\Time \subseteq \RRRp$. In general, a linearly ordered set is sufficient. In the following, we use $\Time$ as a semiring with operations $\oplus = \min$ and $\odot = +.$
\end{defn}

For the operations on temporal networks with zero latency, described in our articles \cite{TQSON,PraprotnikAMC}, we assumed $\Time \subseteq \NNN$.

In a temporal network, the nodes $v \in \vertices{V}$ and the links $\ell \in \edges{L}$ are not necessarily present or active all the time. Let $T(v), \; T \in \functions{P},$ be the set of time points in which the node $v$ is present; and let $T(\ell), \; T \in \functions{W},$ be the set of time points in which the link $\ell$ is active. We require that the following \emph{consistency condition} holds: If a link  $\ell(u,v)$ is active at the time $t$  its end nodes  $u$ and $v$ must be present at the time $t.$ Formally,
\begin{equation}\label{eq:dosl}
T(\ell(u,v)) \subseteq T(u) \cap T(v) .
\end{equation}

\begin{defn}
The static network consisting of links and nodes present in a temporal network at the time $t \in \Time$ is denoted with $\network{N}(t)$ and is called a \emph{time slice} of the temporal network at the time $t.$ 

Let $\Time' \subset \Time$. Time slices are extended to the set $\Time'$ as
\[ \network{N}(\Time') = \bigcup_{t\in \Time'} \network{N}(t) . \]
\end{defn}

When we are interested in walks in temporal networks, there are usually additional information on the links of the network. 

\begin{defn}
	The \emph{latency}  $\tau \in \functions{W}, \; \tau : \edges{L} \times \Time \to \RRR.$ The value of $\tau(\ell,t)$ represents the time needed to traverse the link $\ell$ if the transition is started at the time $t.$ If the latency  $\tau$ is omitted, we  assume $\tau(\ell,t)=0$ for all $\ell \in \edges{L}$ and for all $t \in \Time.$ 
	\end{defn}
	
\begin{defn}
The \emph{weight}  $w \in  \functions{W}, \; w : \edges{L} \times \Time \to \RRRp,$ with values $w(\ell,t)$ representing length, cost, flow, etc.~on the link $\ell$ if the transition is started at the time $t.$  If the weight  $w$ is omitted, we  assume $w(\ell,t)=1$ for all links $\ell\in \mathcal L$ and all times $t\in \Time.$ In some cases the weights are structured.
\end{defn}

\begin{defn}
	A walk in a temporal network is called a \emph{journey}. The journey $\sigma(v_0,v_k,t_0)$ from the start node $v_0$ to the end node $v_k$ with the begining $t_0$ is a finite sequence
$$(t_0,v_0,(t_1,\ell_1),v_1,(t_2,\ell_2),v_2,\dots,v_{k-2},(t_{k-1},\ell_{k-1}),v_{k-1},(t_k,\ell_k),v_k),$$
where $v_i \in \mathcal V, i=0,1,\dots,k,$ and $\ell_i \in \edges{L}, t_i \in \Time, i = 1,2,\dots,k.$ The links have to link the appropriate nodes, $\ell_i(v_{i-1},v_i).$ The triples $v_{i-1},(t_i,\ell_i)$ tell that we started from the node $v_{i-1}$ at the time $t_i$ along the link $\ell_i$.

We denote $t_0' = t_0,\; t_i' = t_i + \tau(\ell_i,t_i), \; i = 1,2,\dots,k.$ These are the times when we arrive at the next node. For a journey  $t_{i-1}' \leq t_i$ has to hold and the link $\ell_i$ has to be present in the time interval $[t_i,t_i']$ for all $i = 1,2,\dots,k.$ Also, the node $v_0$ has to be present at the time $t_0.$
\end{defn}

Note that by the consistency condition it also holds that the nodes $v_{i-1}$ and $v_{i}$ are present in the time interval $[t_i,t_i']$.

\begin{defn}
	A journey is \emph{regular} if the node $v_i$ is present while waiting in the node for the next transition, that is during the time interval  $[t_{i}',t_{i+1}],\; i=0,1,\dots,k-1.$ 
\end{defn}

\begin{defn}
	A journey $\sigma$ has a (graph) \emph{length} equal to the number $k$ of links it contains, $|\sigma| = k.$ The \emph{duration} of the journey is equal to $t(\sigma) = t_k'-t_0$ and the  \emph{value} of the journey is equal to
$$w(\sigma) = w(\ell_1,t_1) \odot w(\ell_2,t_2) \odot \cdots \odot w(\ell_k,t_k) = \bigodot_{(t,\ell) \in \sigma} w(\ell,t)$$
for the multiplication in the appropriate semiring. 
\end{defn}

\begin{defn}The time $t_0$ is the \emph{begining} of the journey, the time $t_1$ is the \emph{departure} and $t_k'$ is the \emph{arrival} (end of the journey). The time $t_k'-t_1$ is called a \emph{strict duration} of the journey. Times $t_{i+1}-t_{i}',\; i = 0,1,\dots,k-1,$ are the \emph{waiting times} of the journey.
\end{defn}

\begin{defn}
A \emph{jump} is a journey inside a given network time slice $\network{N}(t).$ Jumps have zero latency and zero waiting times.
\end{defn}

\begin{defn}
	The \emph{fastest} journey is the one with the smallest strict duration. The \emph{foremost} journey is the one with the smallest arrival time. The \emph{cheapest} journey is the one with the smallest value.
\end{defn}

\begin{defn}
	A part of the journey $\sigma(v_0,v_k,t_0)$ from the node $v_i$ to the node $v_j$ with the beginning at the time $t_i,$ 
	$$(t_i,v_i,(t_{i+1},\ell_{i+1}),v_{i+1},\dots,v_{j-1},(t_j,\ell_j),v_j),$$ is caled a \emph{stage} of the journey. An \emph{ubiquitous foremost journey} is the foremost journey for which every stage is a foremost journey between the nodes $v_i$ and $v_j$ with the beginning $t_i.$
\end{defn}

It has been shown in \cite{Xuan_JoFoCS2002} that if there exists a journey between two nodes, then the ubiquitous foremost journey exists between them.

\subsection{Temporal quantities} \label{sec:TQ}

In temporal networks besides the presence or absence of nodes and links, also the values of node and link properties change through time. For the description of the temporal properties, we introduced \emph{temporal quantities} in \cite{TQSON}. Let $a(t)$ be the value of the property $a$ at the time $t.$ We assume that the values $a(t)$ of the function $a$ belong to the semiring $(A,\oplus,\odot,0,1).$ The node or the link that $a$ is describing is not necessarily present at all times. Therefore the function $a$ is not defined for all values $t \in \Time.$ 
\begin{defn}
Let $(A,\oplus,\odot,0,1)$ be a semiring and let the function $a: T_a \to A$ describe a temporal property in a temporal network. A	\emph{temporal quantity} $\hat a: \Time \to A$  is an extension of the function  $a,$
\[  \hat{a}(t) = \left\{\begin{array}{ll} 
a(t), & t \in T_a, \\
0, & t \in \Time \setminus T_a.
\end{array}\right. \]
\end{defn}

Note that the values of temporal quantities while the node or the link is not present are defined as the zero of the semiring $A.$ This means that the values along the sequential links are equal to 0 (describing nonexistence) if one of the sequential links does not exist.

In the rest of the article, we denote temporal quantities with $a$ instead of $\hat a.$

\subsection{Temporal semirings}

In this section, the latency and the waiting times in the temporal network are equal to zero. We described the temporal semirings in more detail and provided algorithmic support in our articles \cite{TQSON,PraprotnikAMC}.

\begin{defn}\label{def:opCasovniTemp}
Let $A_\Time$ be a set of all temporal quantities over the chosen semiring $(A,\oplus,\odot,0,1)$ for the lifetime  $\Time,$ that is $A_\Time = \{a:\Time \to A\}.$ In the set $A_\Time,$ we define the \emph{addition}
\[ (a \oplus b)(t) = a(t) \oplus b(t) \]
and the \emph{multiplication}
\[ (a \odot b)(t) = a(t) \odot b(t). \]
\end{defn}

The operations on the left hand side operate in the set $A_\Time$ of temporal quantities over the semiring $A$ for the lifetime $\Time,$ and the operations on the right hand side operate in the semiring $A.$

\begin{trd}
The set $A_\Time$ for the operations from the definition \ref{def:opCasovniTemp} is a semiring with the zero $0(t) = 0,\; t \in \Time,$ and the unit $1(t) = 1,\; t \in \Time.$ 
\end{trd}

\begin{proof}
The operations are defined pointwise and the semiring properties in  $A_\Time$ follow from the properties of the semiring $A.$
\end{proof}

\begin{defn} \label{def:casovni}
Let $A$ be a combinatorial (shortest paths, geodetic, etc.) semiring. The semiring $A_\Time$ is called a \emph{temporal combinatorial (shortest paths, geodetic, etc.) semiring}. 
\end{defn}

We can construct a matrix semiring over the temporal semirings. Such matrices can be used to describe temporal networks. Because the values of  $a(t)$ and $b(t)$ in the definition \ref{def:opCasovniTemp} correspond to the same time point $t,$ the latency and the waiting times are restricted to zero for the whole lifetime. The use of this semiring in temporal networks is restricted to jumps and not to arbitrary journeys for the operations to make sense.

\subsection{Semiring of increasing functions}\label{sec:monotonefje}

\begin{defn}
	A function $f$ is  \emph{increasing} iff $f(x) \geq f(y)$ for all $x,y$ of its domain for which $x \geq y.$ We say that a function $f$ is \emph{expanding} if $f(x) \geq x$ for all $x$ of its domain.
\end{defn}

\begin{trd}\label{trd:monotone}
The set 
$$A = \left\{\begin{array}{c}
 	f:\NNN \to \NNN; \mbox{ function }  f \mbox{ is increasing and expanding}\end{array}\right\}$$
 is a semiring for the operations
\begin{align*}
	f \oplus g &= \min (f,g) \quad \mbox{ and } \quad
	f \odot g = g \circ f.
\end{align*}
The zero is a function $f \equiv \infty$ and the unit is the identity function $f = id.$ For the domain or codomain of functions $f$ we could also choose  the sets $\ZZZ, \; \RRR,$ or $\RRRp.$
\end{trd}

\begin{proof}
This semiring is very similar to the semiring from \cite[p. 346, Section 4.2.1]{Gondran:2008:GDS:1386688}.
\end{proof}

\begin{defn}
	The semiring $A$ from theorem \ref{trd:monotone} is called the \emph{semiring of increasing functions.}
\end{defn}

The semiring of increasing functions is complete, idempotent ($\min(f,f) = f$), closed for $f^\star = 1 \oplus f \odot f^\star = \min(id,f^\star\circ f) = id,$ and absorptive  ($\min(id,f) = id$) because $f^\star$ and $f$ are increasing and expanding functions. 

\subsection{First arrival semiring}

We start with an equation, similar to Bellman's equation \eqref{eq:BellmanFord}, for  finding the foremost journeys in a temporal network.

Let a temporal quantity $a_{uv}$ describe the latency along the link $(u,v)$ and let $T(u,v,t_0)$ be the first possible time at which we can arrive at the node $v$ if we start at the node $u$ at the time $t_0.$
Then
$$T(u,u,t_0) = t_0$$
and
\begin{equation}
\label{eq:mincas}
T(u,v,t_0) = \min _{w: (w,v)\in \mathcal L} \left(\min_{t \geq T(u,w,t_0)} (t+a_{wv}(t))\right).
\end{equation}
If we are interested in the duration, we subtract the begining  $t_0$ from the result.

We would like to construct a semiring that gives us this equation, similarly to the way that the shortest paths semiring gives Bellman's equation. The semiring operations are not obvious, as there are three operations (two minimums and the addition) in equation \eqref{eq:mincas}.

What we can see is that it is useful to define a function (temporal quantity) that tells the first arrival time for  the given start, end, and begining of the journey.

From the network interpretation, we can see what the appropriate semiring addition and multiplication are:

Let our journey take two sequential links  $(u,w)$ and $(w,v).$ The first arrival time at the node  $w$ along the link $(u,w)$ is described with the temporal quantity  $f,$ and the first arrival at the node  $v$ along the link $(w,v)$ is described with the temporal quantity $g.$ The corresponding journey is outlined in Figure \ref{fig:minCasKol}.

\begin{figure}[!ht]
	\begin{center}
%
%
%
%
\includegraphics{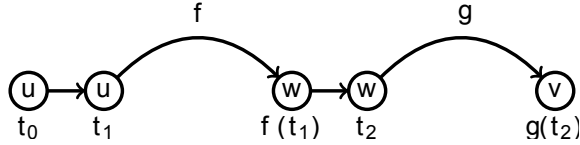}
	\end{center}
	\caption{A journey along sequential links.}
	\label{fig:minCasKol}
\end{figure}

From the begining $t_0$ of the journey, we wait in the node $u$ for some favorable time $t_1$ when we move along the link $(u,w).$ This part of the journey ends at the time $f(t_1).$ Afterwards, we wait for a favorable time $t_2$ in the node  $w.$ At that time, we move along the link $(w,v).$ The journey ends at the time $g(t_2).$ We are interested in the first arrival at the node $v$ if we start at the node $u$ at the time $t_0$ and visit the node $w$ inbetween. That gives us an appropriate semiring multiplication
$$(f \odot g) (t_0) = \min_{t_1\geq t_0 \atop t_2 \geq f(t_1)} g (t_2).$$

We note that if $f$ and $g$ are increasing functions, this equation is equivalent to
$$(f \odot g) (t_0) = g(f(t_0)) = (g \circ f)(t_0).$$
We also point out that the multiplication is not commutative which means that the order in which the links are traversed is important. That is in accordance with our intuition.

When the journey can take us along two parallel links (one possibility is presented in Figure \ref{fig:minCasKolplus}), we start at the time $t_0$ and wait for the time $t_1,$ when it pays to go along the edge for which the arrival times are described with the function $g.$ This journey ends at the time $g(t_1).$ If we wish to take the other link, where the arrival times are described with the function $f,$ we wait for some other time $t_2$ and arrive at $v$ at the time $f(t_2).$ The first arrival time is the smallest of the times  $f(t_2)$ and $g(t_1).$

\begin{figure}[!ht]
	\begin{center}
%
%
%
%
\includegraphics{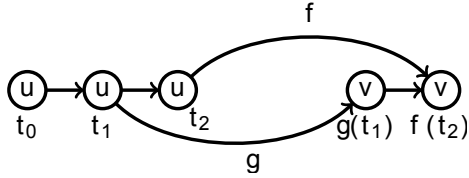}
	\end{center}
	\caption{A journey on parallel links.}
	\label{fig:minCasKolplus}
\end{figure}

That is
$$(f\oplus g)(t_0) = \min_{t_1\geq t_0 \atop t_2 \geq t_0}\big(f(t_2),g(t_1)\big) = \min_{t \geq t_0} \big(f(t),g(t)\big).$$

When $f$ and $g$ are increasing, the equation is equivalent to
$$(f\oplus g)(t_0) = \min (f(t_0),g(t_0)).$$

The two appropriate operations are exactly the ones from the semiring of increasing functions.

Let the values of the temporal quantity $a$ represent the latency along the link. Remember that  $a(t) = \infty$  at times $t \in \Time \setminus T_a.$ We assign a function $f$ to the temporal quantity $a$:
\begin{equation}
\label{eq:delay}
a \mapsto f: \; f(t) = \min_{\tau \geq t} \{\tau + a(\tau)\}.
\end{equation}
The function $f$  is increasing and expanding if  $a \geq 0$ which it usually is as the travel times are nonnegative. If $a$ is describing the latency along the link $(u,v)$ the function $f$  is describing the first arrival time from  $u$ to $v.$ The value $f(t)$ is the first arrival if we begin the journey at the time $t.$

The first arrival times in a temporal network with arbitrary waiting times and given latencies can be computed with the addition and multiplication in the semiring of increasing functions.

\begin{defn}
Let $\mathcal N = (\mathcal V, \mathcal L, \Time, a)$ be a temporal network and let the temporal quantity $a: \Time \to \Time$ describe the latency. We assign a function $f$ to the temporal quantity $a$ as in the equation \eqref{eq:delay}. The semiring
$$\mathbb T = \big(\{f: \Time \to \Time\},\min,\circ,\infty,id\big)$$
is called the \emph{first arrival semiring.}
\end{defn}

\subsubsection{Example temporal network}

We take for the illustration of these principles a simple temporal network with five nodes and seven links, as is shown in Figure \ref{fig:tempNet}. 

\begin{figure}
	\begin{center}
		\includegraphics[width=7cm]{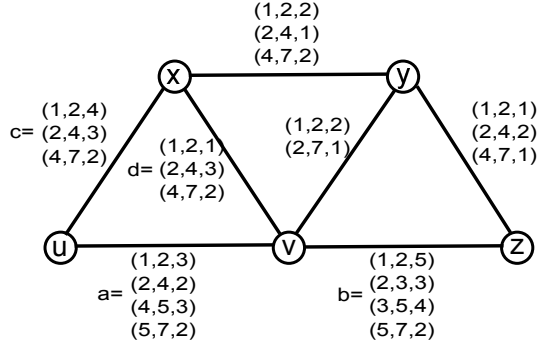}
	\end{center}
	\caption{An example temporal network. The weights on the links are temporal quantities corresponding to latencies/travel times along the links at different times.}
	\label{fig:tempNet}
\end{figure}

In the Python library TQ, we assume $\Time \subset \NN$ and describe temporal quantities in the form $[(s_i,f_i,v_i)],$ which means that on the time interval $[s_i,f_i)$ the value of the temporal quantity is equal to $v_i.$ For a more detailed description see \cite{PraprotnikAMC,TQSON}. In our examples, we use the same notation.

The temporal quantity $a$ is the weight on the edge $\{u,v\}.$ It tells the time needed to get from $u$ to $v$ or from $v$ to $u$ at different time points. Specifically, if we start along the edge $\{u,v\}$ at the time point 1, the time needed to traverse this edge is equal to 3; if we start at times 2 or 3, the time needed is equal to 2; etc.

From the temporal quantity $a,$ we get all possible arrival times along the edge $\{u,v\}$. In the library TQ, we get the results using the function \texttt{arrival}$(a)$ = [4, 4, 5, 7, 7, 8]. From \texttt{arrival}$(b)$ = [6, 5, 7, 8, 7, 8] we see, that it is sometimes better to wait before we begin the journey, as the arrival time with the start at 2 is 5, and the arrival time with the start at 1 is equal to 6. 

From the arrival times, we compute the first arrival times by equation \eqref{eq:delay}. In the library TQ, we use the function \texttt{firstArrival}. In this example, \texttt{firstArrival}$(b)$ = [(1, 3, 5), (3, 6, 7), (6, 7, 8)]. This tells us: if we start the journey at times 1 or 2, the first arrival time is 5; if we start from 3 to 6 (i.e.~at times 3, 4, or 5), the first arrival time is 7; etc.

The operations in the first arrival semiring make sense if the temporal quantities describe the first arrival times.

The sum and the product in the first arrival semiring are implemented as functions \texttt{fAsum} and \texttt{fAprod} in the library TQ. For the example network, we get
\texttt{fAsum} $(\texttt{firstArrival}(a),\texttt{fAprod}(\allowbreak\texttt{firstArrival}(c),\texttt{firstArrival}(d)))$ = [(1, 3, 4), (3, 4, 5), (4, 6, 7)]. This temporal quantity tells us the first arrival times for the journey from $u$ to $v$ if the journey takes us directly over the edge $\{u,v\}$ or across the edges $\{u,x\}$ and $\{x,v\}.$

\subsection{Generalized geodetic semirings}

The generalized geodetic semirings are defined in a very similar way as the geodetic semiring from Section \ref{sec:geodez}.

\begin{defn} \label{def:opPG}
In a set $\Time \times A,$ where $(A,\oplus,\odot,\mathbf 0,\mathbf 1)$ is an arbitrary complete semiring (combinatorial, shortest paths, geodetic, etc.), the operations  \emph{addition} $\boxplus$ and \emph{multiplication} $\boxtimes$ are defined as 
 $$(\tau,a) \boxplus (\sigma,b) = \left(\min(\tau,\sigma), \left\{\begin{array}{ll}
 a, & \tau < \sigma, \\
 a \oplus b, & \tau = \sigma,\\
 b, & \tau > \sigma
 \end{array}\right. \right)$$
 and
 $$(\tau,a) \boxtimes (\sigma,b) = (\tau + \sigma, a \odot b).$$
  \end{defn}
  
 \begin{trd}
 	The set $\Time \times A$ is a semiring for the addition $\boxplus$ and the multiplication $\boxtimes.$  The zero is $(\infty,\mathbf 0)$ and the unit is $(0,\mathbf 1).$
 \end{trd}
 
 \begin{proof}
 	The construction is almost identical to the one for the geodetic semiring and the semiring properties follow in the same way as in  \cite{Batagelj94semiringsfor} from the properties of the operations in  $\Time$ and $A.$
 \end{proof}
 
 \begin{defn}
The semiring $G_{\Time \times A} = \big(\Time \times A, \boxplus, \boxtimes,(\infty,\mathbf 0), (0,\mathbf 1)\big)$ is called a \emph{generalized geodetic semiring.}
 \end{defn}

\subsection{Traveling semirings}

The next question is how to combine different information on the links. For example, latency and the number of ways to traverse it, or latency and distance.

Let the temporal quantity $a: \Time \to \Time$ describe the latency and let the temporal quantity  $i \in A_\Time$ over a chosen semiring $(A,\oplus,\odot,\mathbf 0,\mathbf 1)$ describe some other information about the link.

We want to compute
$$(f, n)(t) = \left(\min_{\tau \geq t}(\tau + a(\tau)), \bigoplus_{\sigma \; \in \;\mbox{\scriptsize{Argmin}}_{\tau \geq t}(\tau + a(\tau)) \atop \sigma \geq t} i(\sigma)\right).$$

The first component $f$ stays the same as in the first arrival semiring (equation \eqref{eq:delay}) and tells the first arrival along the link after the time $t.$ In the second component, $n,$ we sum (over the chosen semiring $A$) the values along the links on which the minimal arrival time is achieved and that start after the time  $t.$

First, we do a simple transformation

\begin{align*}
(a,i) &\mapsto (a',i) \quad \mbox{ where } \quad
a'(t) = t + a(t)
\end{align*}
from which we get
\begin{equation} \label{eq:transf}
(f, n)(t) = \left(\min_{\tau \geq t}a'(\tau),\bigoplus_{\sigma \; \in \;\mbox{\scriptsize{Argmin}}_{\tau \geq t}a'(\tau) \atop \sigma \geq t} i(\sigma)\right).
\end{equation}

The last equation is simplified by summing over the corresponding generalized geodetic semiring $G_{\Time\times A}.$ The equation \eqref{eq:transf} can be rewritten as

\begin{equation}\label{eq:casGeod}
(f, n)(t)= \mbox{\raisebox{-0.9em}{$\overset{\mathlarger{\mathlarger{\mathlarger{\mathlarger{\boxplus}}}}}{\tau \geq t}$}} \; (a'(\tau),i(\tau)).
\end{equation}
Note that $f \in \mathbb T$ and $n \in A_\Time.$

\subsubsection{Example temporal network} \label{example:countArrival}

Again, we consider the temporal network from Figure \ref{fig:tempNet}. The transformation from the equation \eqref{eq:casGeod} for the temporal combinatorial semiring $A_\Time$ is done with the function $\texttt{countArrival}.$ For example, we get \texttt{countArrival}$(a)$ = [(1, 2, (4, 2)), (2, 3, (4, 1)), (3, 4, (5, 1)), (4, 5, (7, 2)), (5, 6, (7, 1)), (6, 7, (8, 1))] and \texttt{countArrival}$(b)$ = [(1, 3, (5, 1)), (3, 4, (7, 2)), (4, 6, (7, 1)), (6, 7, (8, 1))]. The latter tells us that we can get from $v$ to $z$ at times 1 and 2 soonest at the time 5, and there is 1 possible choice (we begin at 2 and finish at 5); at time 3, the first arrival at $z$ is time 7 and there are two possible choices (begin at 3 and finish at 7 or begin at 5 and finish at 7); etc.

\subsubsection{Operations in traveling semirings}

The transformation \eqref{eq:casGeod} of the temporal quantities $a,$ representing latency, and $i,$ representing some other information, returns a pair $(f, n)$ belonging to the set
$$G_A(\Time) = \left\{
(f, n); \; f \in \mathbb T, \; n \in  A_\Time \right\}.$$
 
 \begin{defn}\label{def:potovalOp}
On a set of function pairs $G_A(\Time)$ we define the \emph{addition} $\diamondplus$ and the \emph{multiplication} $\diamonddot$ with

\begin{align*}
\big((f, n) &\diamondplus (g, m)\big)(t) = (f,  n)(t) \boxplus (g, m)(t), \\
\big((f, n) &\diamonddot (g, m)\big)(t) = \big((g \circ f)(t), n(t) \odot   m (f (t))\big).
\end{align*}
The operation $\boxplus$ is the addition in the generalized geodetic semiring  $G_{\Time \times A}$  and the operation $\odot$ is the multiplication in the semiring $A.$
\end{defn}

The definitions can be read as: If there are two parallel links, we choose the one that arrives first and preserve the same additional value. If both parallel links arrive at the same time, we sum the corresponding additional values.

On sequential links, the arrival time is the same as the arrival over the second link. The journey along the second link can begin after the first arrival along the first link (time $f(t)$). The value of the second component is the value on the first link if we start the journey after the time $t$ multiplied by  the value of the second link if we traverse the link after the time  $f (t).$ 

The first component tells the first arrival and the second component tells additional values for the ubiquitous foremost journey, depending on the semiring $A.$ If $A$ is a combinatorial semiring, the second component tells the number of the ubiquitous foremost journeys. If $A$ is the shortest paths semiring, the second component tells the length of the cheapest among the ubiquitous formost journeys.

\begin{trd}
The set $G_A(\Time)$ is a semiring for the operations from the definition \ref{def:potovalOp}.  The zero is a pair of constant functions $(\infty, \mathbf{0}).$ The unit is $(id,\mathbf 1).$ The second component of the unit is a constant function.
\end{trd}

\begin{proof}
	The associativity, commutativity, and the neutral element for the addition follow from the properties of the generalized geodetic semiring.

First, we show that $(id,\mathbf 1)$ is the unit
\begin{align*}
\big(( f,  n)\diamonddot (id,\mathbf 1)\big)(t) &= \big( f(t),  n(t) \odot \mathbf 1\big) = ( f, n)(t),\\
\big((id,\mathbf 1)\diamonddot ( f,  n)\big)(t) &= \big( f(t),\mathbf 1 \odot  n(t)\big) = ( f,  n)(t).
\end{align*}
and that $(\infty, \mathbf{0})$ is the zero
\begin{align*}
\big(( f,  n) \diamonddot (\infty, \mathbf 0)\big)(t) &= (\infty,  n(t) \odot \mathbf 0) = (\infty, \mathbf 0), \\
\big((\infty,\mathbf 0) \diamonddot (f, n)\big)(t) &= ( f(\infty), \mathbf 0\odot n(\infty)\big) = (\infty,\mathbf 0), \mbox{ because $ f$ is expanding.}
\end{align*}

Now check the multiplication associativity and the distributivity.
First the associativity:
{\small
\begin{align*}
\big(\big(( f, n)\diamonddot ( g, m)\big)\diamonddot ( h,  r)\big)(t) &= \big(( g \circ  f)(t), n(t) \odot  m ( f(t))\big) \diamonddot ( h(t), r(t))\\
&=\big(( h\circ  g \circ  f)(t),  n(t) \odot  m( f(t))\odot  r(( g \circ  f )(t))\big) \\
\big(( f, n)\diamonddot \big(( g, m)\diamonddot ( h,  r)\big)\big)(t) &= ( f, n)(t) \diamonddot \big(( h \circ  g)(t), m(t) \odot  r( g(t))\big)\\
&= \big(( h \circ  g \circ  f)(t),  n(t) \odot  m( f(t))\odot  r( g ( f (t)))\big).
\end{align*}}
We get the same result in both cases, therefore the associativity holds. Check for distributivity:
\begin{align*}
(( h, r) \diamonddot ( f,  n))(t) = (( f \circ  h)(t)&, r(t) \odot  n( h(t))),\\
(( h, r) \diamonddot ( g,  m))(t) = (( g \circ  h)(t)&, r(t) \odot  m( h(t)))
\end{align*}
and
$$
(( h, r) \diamonddot ( f,  n) \diamondplus ( h, r) \diamonddot ( g,  m))(t)  = $$
$$\left(\min ( f( h (t)), g ( h(t))),
\left\{ \begin{array}{l l}
 r(t) \odot  n( h(t)), &  f ( h (t)) <  g ( h (t)) \\
 r(t) \odot ( n ( h(t))\oplus m( h(t))), &  f ( h (t)) =  g ( h(t)) \\
 r(t) \odot  m( h(t)), &  f ( h (t)) >  g ( h (t)) \\
\end{array}\right.\right).$$
We used the distributivity of the semiring $A.$
The other side of the distributivity equation gives
$$(( f,  n) \diamondplus ( g,  m))(t) = \left(\min( f(t), g(t)),\left\{\begin{array}{ll} 
 n(t),& f(t) <  g(t)\\
( n \oplus  m)(t),&  f(t) =  g(t)\\
 m(t), &  f(t) >  g(t)
\end{array}\right.\right),$$
which we multiply from the left
$( h,  r)(t) \diamonddot$ and get
$$\left((\min( f,  g)\circ  h)(t), r (t) \odot \left\{\begin{array}{ll} 
 n( h(t)),& f( h(t)) <  g( h(t))\\
( n \oplus  m)( h(t)),&  f( h(t)) =  g( h(t))\\
 m( h(t)), &  f( h(t)) >  g( h(t))
\end{array}\right.\right).$$
So the left distributivity holds. If we multiply $(( f,  n) \diamondplus ( g,  m))(t)$ on the right hand side $\diamonddot ( h,  r)(t)$ we get
$$\left(( h \circ \min( f,  g))(t),\left\{\begin{array}{ll} 
 n(t),& f(t) <  g(t)\\
( n \oplus  m)(t),&  f(t) =  g(t)\\
 m(t), &  f(t) >  g(t)
\end{array}\right. \odot  r(\min( f(t), g(t)))\right),$$
which is the same as the results of the following computations
\begin{align*}
(( f,  n) \diamonddot( h,  r))(t) &= (( h \circ  f)(t), n(t) \odot  r( f(t))),\\
(( g,  m) \diamonddot( h,  r))(t) &= (( h \circ  g)(t), m(t) \odot  r( g(t))),
\end{align*}
which adds with $\diamondplus$ to
{\small
$$\left(\min( h ( f(t)), h( g(t))),\left\{\begin{array}{ll}
 n(t) \odot  r( f(t)),& h( f(t)) <  h ( g(t)),\\
 n(t) \odot  r( f(t)) \oplus  m(t)\odot  r( g(t)),& h( f(t)) =  h ( g(t)),\\
 m(t) \odot  r( g(t)),& h( f(t)) >  h( g(t))
\end{array}\right. \right).$$}
The right distributivity holds, as  $f,g$ and $h$ are increasing and the semiring $A$ is distributive.

The distributivity holds and $G_A(\Time)$ is a semiring.
\end{proof}

\begin{defn}
	Let $A$ be a combinatorial (shortest paths, geodetic, etc.) semiring. The semiring
	$$\left(G_A(\Time), \diamondplus, \diamonddot, (\infty, \mathbf{0}), (id,\mathbf 1)\right)$$	
	is called the \emph{traveling combinatorial (shortest paths, geodetic, etc.)} semiring.
\end{defn}

\subsubsection{Example temporal network} \label{example:tCop}

We continue the example from Figure \ref{fig:tempNet}. The traveling combinatorial semiring operations are implemented as functions \texttt{tCsum} and \texttt{tCprod}. Both operations are used with temporal quantities with values corresponding to pairs (first arrival time, number of possible ways of first arrivals) which we get from latencies with the function \texttt{countArrival}, as was shown in Section \ref{example:countArrival}. The results are \texttt{tCsum} $(\texttt{count}\allowbreak \texttt{Arrival}(a), \texttt{countArrival}(b))$ = [(1, 2, (4, 2)), (2, 3, (4, 1)), (3, 4, (5, 1)), (4, 5, (7, 3)), (5, 6, (7, 2)), (6, 7, (8, 2))] and
\texttt{tCprod}$(\texttt{countArrival}(a), \allowbreak \texttt{count}\allowbreak\texttt{Arrival}(b))$ = [(1, 2, (7, 2)), (2, 4, (7, 1))]. The latter tells us, that to get from $u$ to $z$ via $v$, the first arrival time is 7, and that if we begin the journey at time 1 there are 2 possibilities, if we start at times 2 or 3, there is one possible journey. If we begin the journey later, there is no way to get to $z$ during the network lifetime.

In more detail, at the time 1 there are two ubiquitous foremost journeys. First, at the time 1 we start at $u$ along the edge $\{u,v\}$ which takes us 3 time units. We arrive at $v$ at the time 4 and wait till 5 to cross the edge $\{v,z\}.$ This takes us 2 time units. The arrival time at $z$ is 7. The second ubiquitous foremost journey is, we start at the time 2 in $u$ along the edge $\{u,v\}$ which takes us 2 time units. We arrive at $v$ at the time 4 and wait till 5 to cross the edge $\{v,z\}.$ This takes us 2 time units and we arrive at $z$ at 7. 

Note that there is a third possibility for a foremost journey: We take the edge $\{u,v\}$ at the time 3 which takes 2 time units and arrive at $v$ at the time 5. We cross the edge $\{v,z\}$ at 5 and again arrive at $z$ at the time 7. This is not an ubiquitous foremost journey from $u$ to $z$ because the stage from $u$ to $v$ is not a foremost journey as it does not finish at the time 4.

We compute the result for two other possible routes from $u$ to $z$: $u \to v \to y \to z$ gets us [(1, 2, (6, 2)), (2, 3, (6, 1)), (3, 4, (7, 1))], that is at the time 1 there are two possible ubiquitous foremost journeys with arrival time 6; at time 2 there is one such journey, and at time 3 there is one journey that arrives at the time 7.

The second journey is $u \to x \to v \to z$ which gets us []. That means that there is no way to take this route and finish in the network lifetime.

In temporal networks, it is not generally true that the foremost journey includes only foremost stages which holds for shortest paths in static networks. See Figure \ref{fig:casGeod} as an example. The weights on links are the latencies and the number of ways to cross them. The latency on the link $(u,v)$ is 2 at the time point 1 and  3 at the time point 2. Between the nodes $v$ and $w$ the latency is equal to 2 at the time point 5. Outside the specified times the links are not present.

There are $k$ foremost journeys between the nodes $u$ and $v$ that have the arrival time 3. Between the nodes  $v$ and $w$ there are $n$ foremost journeys. Between the nodes $u$ and $w$ there are $(m+k)\cdot n$ foremost journeys. Our intuition does not distinguish between waiting in the node $v$ and traveling along a link. The traveling semiring does. The link $(u,v)$ with the weight $(3,m)$ is not taken into account in the semiring as it is not included among the ubiquitous foremost journeys between $u$ and $w.$ We pointed out this shortcoming in the example above.

\begin{figure}[!ht]
	\begin{center}
		\includegraphics{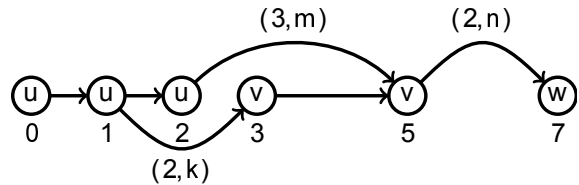}
	\end{center}
	\caption{The foremost journey does not necessarily include only the foremost stages.}
	\label{fig:casGeod}
\end{figure}

\section{Betweenness centrality} \label{sec:between}

Determining important nodes in the network is one of the basic network analysis tools. A lot of different node centralities have been defined for static networks \cite{wasserman1994social}. One of the classical centralitity measures is the betweenness centrality \cite{freeman77,Freeman78centralityin}.

\begin{defn} 
	The \emph{betweenness} of a node $v$ in a network $\mathcal N = (\mathcal V, \mathcal L, \mathcal W)$ is defined with
	$$  b(v) = \frac{1}{(n-1)(n-2)} \sum_{u,w \in \mathcal V \atop |\{v,u,w\}| = 3}
	\frac{n_{uw}(v)}{n_{uw}},$$
	where $n_{uw}$ is the number of the shortest paths from $u$ to $w$ and $n_{uw}(v)$ is the number of the shortest paths from $u$ to $w$ that include the node $v.$ If $n_{uw} = 0$ we define $n_{uw}(v)/n_{uw} = 0.$
\end{defn}

The betweenness centrality is based on the shortest paths in the network. The ratio $n_{uw}(v)/n_{uw}$ can be seen as the probability that the communication between $u$ and $w$ goes through  $v.$ Therefore, the betweenness centrality implicitly assumes that all the communication between the nodes of the network takes place only along the shortest paths. That is not necessarily the case and it is a known disadvantage of the betweenness centrality.

%
%

Another possible interpretation of the betweenness centrality of a node $v$ is: Is the difference in the number of shortest paths between pairs of nodes in the network if we exclude the node $v$ from the network big or small? If there is a small change in the number of paths, the node $v$ is not important.

The betweenness centrality is motivated by network traffic monitoring. Which node has the most potential for influencing, security, connectivity, negotiations. It measures the strategic position of nodes.

In \cite{TQSON}, we described the generalization of the betweenness centrality for temporal networks with zero latency. In this article, we aim to generalize it to networks with given latencies and arbitrary waiting times.

\subsection{First arrival betweenness in temporal networks}

We will use the traveling combinatorial semiring $G_A(\Time)$ to define and compute the betweenness in temporal networks. In this semiring, the  pairs of temporal quantities $(f,m)$ are viewed as the first arrival times, $f,$ and as the number of possible traversals of links that result in the first arrival, $m$.

\begin{defn}
	We define the \emph{first arrival betweenness} with respect to the ubiquitous foremost journeys after the chosen time point  $t$ as
$$ \mathbf b_v(t) = \frac{1}{(n-1)(n-2)} \sum_{u,w \in \mathcal  V \atop |\{v,u,w\}| = 3}
\frac{n_{uw(v)}(t)}{n_{uw}(t)}.$$
The $n_{uw}(t)$ denotes the number of ubiquitous foremost journeys from  $u$ to $w$ that begin after the time $t$ and the $n_{uw(v)}(t)$ denotes the number of ubiquitous foremost journeys from $u$ to $w$ that go through $v$ and begin after the time $t.$ If $n_{uw}(t) = 0,$ we omit the corresponding term.
\end{defn}

We point out that this definition has the same problem as the betweenness for static network. It assumes that all the communication / traffic in the temporal network travels along the ubiquitous foremost journeys.

There is another shortcoming to this definition: If the presence of links is sparse, meaning that the links are present only at very few time points, the probability of different journeys having the same finish time is very small. In this case, the betweenness of the nodes is almost always equal to zero as there are very few foremost journeys that end at the same time. For example, think of the network of bus schedules: the time that is needed to get from A to B is rarely the same as the time needed to get from A to B through C. 

We compute the values $n_{uw}(t)$ and $n_{uw(v)}(t)$ from the closure  $\mathbf B$ of a temporal  network matrix over the traveling combinatorial semiring in a similar way as for the static case. The matrix $\mathbf B$ consists of temporal quantities with values of pairs $\big(f_{uv}(t),n_{uv}(t)\big).$ The value $f_{uv}(t)$ is the first arrival time for journeys from $u$ to $v$ with the begining  $t.$ The value $n_{uv}(t)$ tells the number of the ubiquitous foremost journeys begining at the time $t,$ starting at $u,$ and arriving at $v$ at the time $f_{uv}(t).$ 

Once we know the matrix $\mathbf B,$ we compute 
$$n_{uw(v)}(t) = n_{uv}(t) \cdot n_{vw}(f_{uv}(t))$$
if $f_{uw}(t) = f_{vw}(f_{uv}(t)).$ Otherwise $n_{uw(v)}(t)$ is equal to $(\infty, 0).$

\subsubsection{Example temporal network}

To compute the first arrival betweenness, we first implemented the appropriate closure (function \texttt{tempClosure}), and used the traveling semiring operations. The first arrival betweenness operation is implemented as the function \texttt{tempBetween}. For the example network in Figure \ref{fig:tempNet}, the result is written in Table \ref{tab:FAB}.

\begin{table}[!ht] 
	\begin{center}
\begin{tabular}{|r|l|}
	\hline
	$u$ & [] \\
	$v$ & [(1, 2, 0.4802), (2, 3, 0.33332), (3, 4, 0.25), (4, 5, 0.1667), (5, 6, 0.0833)]\\
	$z$ & []\\
	$y$ & [(1, 2, 0.4762), (2, 3, 0.5556), (3, 4, 0.5694), (4, 5, 0.4167), (5, 6, 0.1667)]\\
	$x$ & [(1, 2, 0.0516), (3, 5, 0.0833)]\\
	\hline
\end{tabular}
\caption{First arrival betweenness for the temporal network in Figure \ref{fig:tempNet}.}
\label{tab:FAB}
\end{center}
\end{table} 

The results tell us that the nodes $u$ and $z$ are not important fot the ubiquitous foremost journeys. Throughout the lifetime of the network, the most important nodes are $v$ and $y$, but the relative importance changes. This is logical if we look at the network, as the latencies on the edges adjacent to these two nodes are smaller than the latencies of the edges adjacent to $u$ and $z.$ We also note that the sum of the values is not equal to 1 at later times, as there are not a lot of possible foremost journeys as we approach the network lifetime, and the normalization factor stays the same.

\begin{figure}[!ht]
	\begin{center}
		\includegraphics[width=6cm]{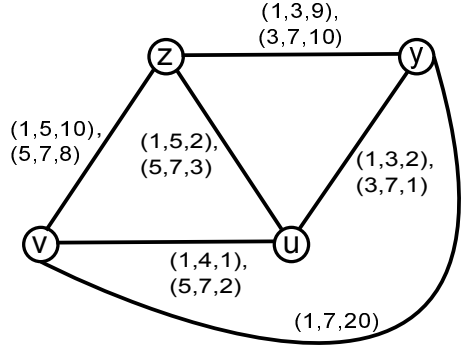}
	\end{center}
	\caption{A second example of the temporal network.}
	\label{fig:tempNet2}
\end{figure}

We compute the first arrival betweenness also for the network in Figure \ref{fig:tempNet2}. It is obvious by the choice of the latencies that the node $u$ is the most important for foremost journeys. We get the result which corroborates this intuition. First arrival betweenness of the node $u$ is equal to [(1, 4, 1.0), (4, 5, 0.5), (5, 6, 0.1667)]. All other nodes have a temporal betweenness of []. So they are not important with respect to ubiquitous foremost journeys.

Again, note the problem with the normalization factor. For static networks, the normalization factor is chosen as the number of all possible paths. But this number is hard to compute for temporal networks. For now, we leave it the same as for static networks, but in the future, it seems that the normalization factor should also depend on time.

\subsection{First arrival betweenness with exclusion}

As we mentioned above, the first arrival betweenness is not a very good indicator of the importance of nodes when the link presence is sparse. For these networks, we propose a definition of importance that is not so strict and tells the node's importance for the whole network lifetime.

\begin{defn}
	The \emph{first arrival betweenness with exclusion} of the node $v$ with respect to the ubiquitous foremost journeys is defined as
	$$ \mathbf b_ {\mathrm{excl}}(v) = \frac{1}{(n-1)(n-2)}\; \sum_{u,w \in \mathcal  V \atop |\{v,u,w\}| = 3} \frac{n_{uw}(v)[\mathcal N]}{n_{uw}(v)[\mathcal N] + n_{uw}[\mathcal N \backslash \{v\}]}.$$
	If $n_{uw}(v)[\mathcal N] + n_{uw}[\mathcal N \backslash \{v\}] = 0,$ the corresponding term is omitted in the computation.
	  
	The $n_{uw}(v)[\mathcal N]$ denotes the number of ubiquitous foremost journeys in the network $\mathcal N$ from node $u$ to node $w$ that include the node $v.$ The $n_{uw}[\mathcal N\backslash \{v\}]$ denotes the number of ubiquitous foremost journeys from $u$ to $w$ in the network $\mathcal N\backslash \{v\}.$  
\end{defn}

The idea behind this definition is simple. We determine the number of ubiquitous foremost journeys from $u$ to $w$ that exist in the network $\mathcal N \backslash \{v\},$ that is the number $n_{uw}[\mathcal N \backslash \{v\}].$ We add the node $v$ to this network (resulting in the network $\mathcal N$) and count the number of ``new'' ubiquitous foremost journeys, i.e.~the journeys that go from $u$ to $w$ through $v,$ denoted with $n_{uw}(v)[\mathcal N].$ Note that these ``new'' journeys can be faster than the old ones, which means that possibly some of the journeys counted in $n_{uw}[\mathcal N \backslash \{v\}]$ are not ubiquitous foremost journeys for the network $\mathcal N.$ This is one of the main differences between the definitions of the first arrival betweenness and the first arrival betweennes with exclusion.

Another big difference is that the first arrival betweenness is a temporal quantity and tells how the node's importance changes through time and the first arrival betweenness with exclusion is a time independent measure.

When the node $v$ is important with respect to ubiquitous foremost journeys in $\mathcal N,$ the value of $n_{uw}(v)[\mathcal N]$ is large and the value of $n_{uw}[\mathcal N \backslash \{v\}]$ is small. This means that there are a lot of journeys through $v$ and few journeys that take other routes. If all the possible routes include $v,$ the ratio for a combination of three different nodes is equal to 1.

The normalization factor is determined in the same way as for static networks. There is no problems with it because the first arrival betweenness with exclusion is not a temporal quantity. Therefore, the factor depends only on the size of the network.

Note that the values of the first arrival betweenness with exclusion are between 0 and 1. A high value of $\mathbf b_{\mathrm{excl}}(v)$ means that the node $v$ is important. In the next Section, we give some numeric examples.

\subsubsection{Examples of the first arrival betweenness with exclusion} \label{sec:FABeEx}

We start with a detailed description of the first arrival betweenness with exclusion on the sparse link presence network with 3 nodes and 3 links that is drawn in Figure \ref{fig:FABexExample1}. The weights on the links represent time points (written in TQ notation) and latency. For example, the edge $\{u,v\}$ is present at the time 2 when it takes 3 time units to cross it, and at the time 7 when it takes 2 time units to cross it.

\begin{figure}[!ht] 
	\begin{center}		
		\includegraphics[width=5cm]{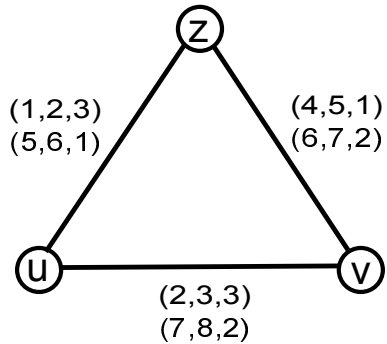}
	\end{center}
	\caption{The  first arrival betweenness with exclusion in a simple network.}
	\label{fig:FABexExample1}
\end{figure}

Because the network is small, we can examine all the possible journeys by hand. The results are written in Table \ref{tab:FABexE1}.

\begin{table}[!ht] 
	\begin{center}
		\begin{tabular}{|c|c||c||c|c|}
			\hline
			\multicolumn{2}{|l||}{Journeys from $u$ to $v$} & &\multicolumn{2}{l|}{Journeys from $v$ to $u$}\\
			\hline
			start & finish & & start & finish \\
			\hline
			2 & 5 & direct &			2 & 5 	\\
			7 & 9 & journeys &			7 & 9   \\	
			\hline
			1 & 5 & going &		4 & 6 \\
			5 & 8 & through $z$ &					& \\	
			\hline
			\multicolumn{5}{l}{\rule{0pt}{-4mm}}\\
			\hline
			\multicolumn{2}{|l||}{Journeys from $u$ to $z$} & &\multicolumn{2}{l|}{Journeys from $z$ to $u$}\\
			\hline
			start & finish & & start & finish \\
			\hline
			1 & 4 & direct &			1 & 4 	\\
			5 & 6 & journeys &			5 & 6   \\	
			\hline
			2 & 8 & going &				4 & 9 \\
			1 & 8 & through $v$ &		& \\	
			\hline
			\multicolumn{5}{l}{\rule{0pt}{-4mm}}\\
			\hline
			\multicolumn{2}{|l||}{Journeys from $v$ to $z$} & &\multicolumn{2}{l|}{Journeys from $z$ to $v$}\\
			\hline
			start & finish & & start & finish \\
			\hline
			4 & 5 & direct &			4 & 5 	\\
			6 & 8 & journeys &			6 & 8   \\	
			\hline
			2 & 6 & going &				1 & 8 \\
			 &  & through $u$ &	5	& 9 \\	
			\hline
		\end{tabular}
		\caption{Possible journeys in the temporal network from Figure \ref{fig:FABexExample1}.}
		\label{tab:FABexE1}
	\end{center}
\end{table}

First, we compute the importance of $z$ for the journeys from $u$ to $v.$ There are 2 ubiquitous foremost journeys in $\mathcal N \backslash \{z\}$ and there are another 2 ubiquitous foremost journeys when we add $z.$ The importance of $z$ for the journeys from $u$ to $v$ is equal to $\frac{2}{4}.$

Now, we look at the journeys from $v$ to $u.$ There are 2 ubiquitous foremost journeys in $\mathcal N \backslash \{z\}$ and there is one more ubiquitous foremost journey when we add $z.$ The importance of $z$ for the journeys from $v$ to $u$ is equal to $\frac{1}{3}.$

The first arrival betweenness with exclusion of the node $z$ is equal to
$$\mathbf b_ {\mathrm{excl}}(z) = \frac{1}{2 \cdot 1} \left(\frac{2}{4} + \frac{1}{3}\right) =\frac{5}{12}.$$

We compute the values of the other two nodes in the same way. Looking at the Table \ref{tab:FABexE1}, we get
\begin{align*}
\mathbf b_ {\mathrm{excl}}(v) &= \frac{1}{2 \cdot 1} \left(\frac{2}{4} + \frac{1}{3}\right) =\frac{5}{12} \quad \mbox{and} \quad 
\mathbf b_ {\mathrm{excl}}(u) = \frac{1}{2 \cdot 1} \left(\frac{1}{3} + \frac{2}{4}\right) =\frac{5}{12}.
\end{align*}

In this example, all the nodes are equally important. That is not surprising as the network is a triangle and the weights are very similar. The first arrival betweenness (without exclusion) is non-zero only for the node $z$ and is equal to $[(1, 2, 0.25), (3, 5,\allowbreak 1.0),\allowbreak (5, 6, 0.5)].$

We implemented this procedure in the library TQ. We compute the first arrival betweenness with exclusion using the function \texttt{betweenTimeEx}.

We test this function on two other small examples in which the underlying graph is a star with 6 nodes (Figure \ref{fig:FABexExample2}). In the first example, the weights of all the links are equal to $(1,10,5).$ In this case, the value of the first arrival betweenness with exclusion of the node $u$ is equal to 1. All other nodes have the value equal to 0. We get the same result without exclusion: The only existing value is the value of the node $u$ which is equal to $[(1, 5, 1.0)].$

In the second example, the weights on the links are as in Figure \ref{fig:FABexExample2}. In this case, the node $u$ has the value of the first arrival betweenness with exclusion of 0.35. The node $v$ has a value of 0.0625. And the node $z$ has a value of 0.0417. Other nodes have the value 0. This happens because there are very few journeys available in this network. For the same reason the sum of all the values is not equal to one. Note that the center of the star still gets the highest value of the first arrival betweenness with exclusion and that the temporal dimension changes the results significally. Vertices $v$ and $z$ are more important than the other 3 periferal nodes because there exist journeys $u \to v \to u$ and $u \to z \to u.$ An application of this would be: if we are waiting for a plane at the site $u,$ can we get to the site $v$ and back before the plane leaves or not? If we can, the site $v$ is more important than the site we cannot visit.

\begin{figure}[!ht] 
	\begin{center}		
		\includegraphics[width=5cm]{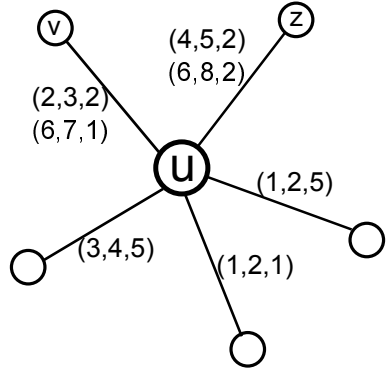}
	\end{center}
	\caption{The  first arrival betweenness with exclusion in a star network.}
	\label{fig:FABexExample2}
\end{figure}

We also list the results of the first arrival betweenness without exclusion for this star network. For the node $u$ we have $[(1, 2, 0.5325), (2, 3, 0.2308), (3, 4, 0.1812), (4, 5, \allowbreak 0.1943), (5, 6, 0.0875),  (6,\allowbreak 7, 0.1)],$ for the node $v$ we have $[(1, 2, 0.0375), (2, \allowbreak 3, \linebreak 0.05513),\allowbreak (3, 4, 0.03125), (4, 5, 0.04286),\allowbreak (5, 6,\allowbreak 0.0375), (6, 7, 0.05)],$ and for the node $z,$ the result is $[(1, 2, 0.02643), (2, 3, 0.03846), (3, 4, 0.025), \allowbreak (4, 5,\allowbreak 0.04286)].$ All the other nodes betweennes is constantly zero. These results also show that $u$ is the most important node in this network. Its importance diminishes when we approach the network lifetime. The other two important nodes are $v$ and $z$ and their importance is very low. This is also in accordance with the results of the betweenness with exclusion.

The main differences between the two definitions are (a) for sparse link presence networks, the betweenness without exclusion is rarely non-zero and therefore not viable, and (b) the betweenness without exclusion gives temporal results which show some changes in importance through time and give a less distinct sense of node importance for the whole lifetime. If we are interested in the overall importance, the betweenness with exclusion is the more suitable one.

\subsection{The importance of selected bus stops in Ljubljana, Slovenia}

From the bus schedules for Ljubljana, Slovenia, we created a temporal network. Because there is a lot of data for the entire city, we chose only a part of the whole network that we know well. This subnetwork consists of 25 bus stops and represents the bus schedule for the selected routes going in one direction from 8 a.m.~untill noon. We chose the routes we know well in order to compare the results with our personal experience.

The results we got by computing the first arrival betweenness with exclusion were in accordance with our intuition -- the least important nodes of the bus network were the last stops of each line. The most important bus stops were the ones where a few lines come together to the same road. The numerical results of the first arrival betweenness are depicted in Figure \ref{fig:LPP}. All the links are directed and are pointing right / down. The nodes of the network are numbered and the numbers next to the nodes represent the values of the first arrival betweenness with exclusions. The nodes without numbers next to them have the first arrival betweenness with exclusions equal to 0.

\begin{figure}[!ht] 
	\begin{center}		
		\includegraphics[width=0.9\textwidth]{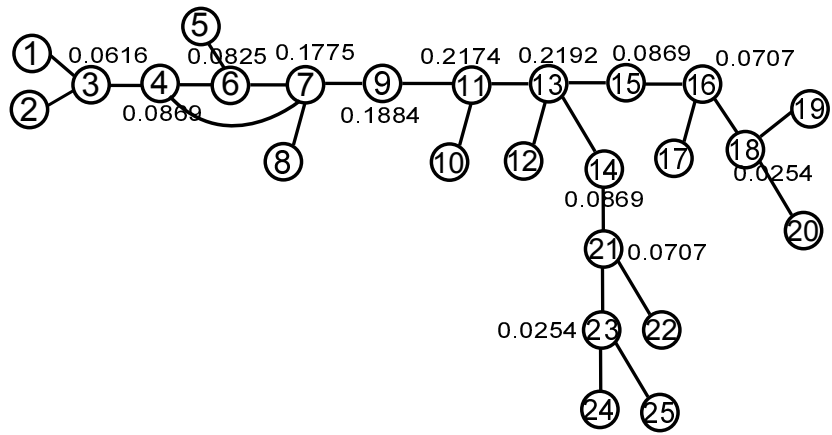}
	\end{center}
	\caption{The  first arrival betweenness with exclusion on a part of the bus schedule net\-work.}
	\label{fig:LPP}
\end{figure}

The nodes that receive the highest values of the first arrival betweenness with exclusions are ``in the middle'' of the network as we can see from Figure \ref{fig:LPP}. We expected such a result from our experience. The nodes with the highest values are nodes 13 (value 0.2192), 11 (value 0.2174), 9 (value 0.1884), and 7 (value 0.1775). All other nodes have a very small or zero value of betweenness. Note that one could expect the values of nodes 16 or 18 to be higer. They are low because it takes quite a long time to get from node 18 to node 13.

We also computed the first arrival betweenness which gave similar but longer results, which are omitted for the sake of brevity.

\section{Conclusion and future work} \label{sec:conclusion}

In the article, we described a new algebraic approach to the analysis of temporal networks that is based on temporal quantities over the selected semiring. We defined a new semiring for computing foremost journeys (first arrival semiring) and traveling semirings in which we can use additional data on the links, besides the latency. 

Our description of a temporal network avoids an explicit record of node and link presence as it is done in most of the literature. We describe the absence implicitly using the zero in the semiring. Our approach allows a wider variety of temporal data to be added to the nodes and to the links of the network. In addition to the latency, it is possible to add lengths, number of ways, and other temporal information. With the definition of the traveling semiring, we can mathematically describe journeys in temporal networks and allow more data in their analysis.

The procedures for the analysis of temporal networks with zero latency and zero waiting times from our articles \cite{TQSON,PraprotnikAMC} and the procedures used in this article are available as a Python library TQ (Temporal Quantities) at \\ \url{http://vladowiki.fmf.uni-lj.si/doku.php?id=tq}. 

We defined two betweenness centralities with respect to the ubiquitous foremost journeys in temporal networks, and showed how to use the semiring operations to compute them. We extended the library TQ to include these routines and tested it on a few examples. We get the results we expected from our knowledge of the real network.

For future research, other methods from static networks could be generalized and special methods that are adapted to the time dimension should be developed. Also, the definition of betweenness could be generalized or adapted in another way that would be more suitable for different data. It seems that the traveling semiring could be used for many different purposes. A new semiring could be constructed that could take into account all the foremost journeys. The normalization factor for the first arrival betweenness should be improved.

There are still questions about the journeys with zero or fixed waiting times. Both cases raise some interesting questions. The fixed time is a very strong assumption and it will be difficult to solve. If a semiring could be constructed for this case, the temporal ``shortest path'' problem could be solved by the matrix closure operation in polynomial time. Since shortest path problem with zero waiting times is NP-hard, this would be highly unlikely.

In the future, we intend to extend the library TQ and provide a better, friendlier version of the program so that it could be used by other researchers.

\section*{Funding}

This work was supported in part by the ARRS, Slovenia, research program P1-0294 and research projects J5-5537 and J1-5433, as well as by a grant within the EURO-CORES Programme EUROGIGA (project GReGAS) of the European Science Foundation.

\bibliographystyle{plain}
\bibliography{semiring}

\begin{thebibliography}{10}

\bibitem{doi:10.2200/S00245ED1V01Y201001CNT003}
John~S. Baras and George Theodorakopoulos.
\newblock Path problems in networks.
\newblock {\em Synthesis Lectures on Communication Networks}, 3(1):1--77, 2010.

\bibitem{Batagelj94semiringsfor}
Vladimir Batagelj.
\newblock Semirings for social network analysis.
\newblock {\em Journal of Mathematical Sociology}, 19(1):53--68, 1994.

\bibitem{TQSON}
Vladimir Batagelj and Selena Praprotnik.
\newblock An algebraic approach to temporal network analysis.
\newblock {\em Submitted to Social Networks}, 2014.

\bibitem{Transp}
Michael G.~H. Bell and Yasunori Iida.
\newblock {\em Transportation network analysis}.
\newblock Chichester: Wiley, 1997.

\bibitem{SCC}
Sandeep Bhadra and Afonso Ferreira.
\newblock Complexity of connected components in evolving graphs and the
  computation of multicast trees in dynamic networks.
\newblock In Samuel Pierre, Michel Barbeau, and Evangelos Kranakis, editors,
  {\em ADHOC-NOW}, volume 2865 of {\em Lecture Notes in Computer Science},
  pages 259--270. Springer, 2003.

\bibitem{GaN}
Bernard Carre.
\newblock {\em Graphs and networks}.
\newblock Clarendon Press; Oxford University Press Oxford; New York, 1979.

\bibitem{TVGsur}
Arnaud Casteigts, Paola Flocchini, Walter Quattrociocchi, and Nicola Santoro.
\newblock Time-varying graphs and dynamic networks.
\newblock {\em International Journal of Parallel, Emergent and Distributed
  Systems}, 27(5):387--408, 2012.

\bibitem{Wardrop}
Jos{\'e}~R. Correa and Nicol{\'a}s~E. Stier-Moses.
\newblock Wardrop equilibria.
\newblock {\em Wiley Encyclopedia of Operations Research and Management
  Science}, 2011.

\bibitem{Dolan:2013:FSF:2544174.2500613}
Stephen Dolan.
\newblock Fun with semirings: A functional pearl on the abuse of linear
  algebra.
\newblock {\em SIGPLAN Not.}, 48(9):101--110, September 2013.

\bibitem{Fletcher:1980:MGA:358876.358884}
John~G. Fletcher.
\newblock A more general algorithm for computing closed semiring costs between
  vertices of a directed graph.
\newblock {\em Commun. ACM}, 23(6):350--351, June 1980.

\bibitem{freeman77}
Linton~C. Freeman.
\newblock A set of measures of centrality based on betweenness.
\newblock {\em Sociometry}, 40(1):35--41, 1977.

\bibitem{Freeman78centralityin}
Linton~C. Freeman.
\newblock Centrality in social networks; {C}onceptual clarification.
\newblock {\em Social Networks}, 1(3):215--239, 1978.

\bibitem{series/sbcs/GeorgeK13}
Betsy George and Sangho Kim.
\newblock {\em Spatio-temporal Networks; Modeling and Algorithms.}
\newblock Springer Briefs in Computer Science. Springer, 2013.

\bibitem{Gondran:2008:GDS:1386688}
Michel Gondran and Michel Minoux.
\newblock {\em Graphs, Dioids and Semirings: New Models and Algorithms
  (Operations Research/Computer Science Interfaces Series)}.
\newblock Springer Publishing Company, Incorporated, 1 edition, 2008.

\bibitem{HolmeRev}
Petter Holme.
\newblock Modern temporal network theory: a colloquium.
\newblock {\em Eur. Phys. J. B}, 88:234.

\bibitem{Holme_PRE2005}
Petter Holme.
\newblock {Network reachability of real-world contact sequences}.
\newblock {\em Physical Review E (Statistical, Nonlinear, and Soft Matter
  Physics)}, 71(4):46119, 2005.

\bibitem{TNsur}
Petter Holme and Jari Saram\"{a}ki.
\newblock Temporal networks.
\newblock {\em Physics Reports}, 519(3):97--125, 2012.

\bibitem{TNbook}
Petter Holme and Jari Saram\"{a}ki.
\newblock {\em Temporal networks. {U}nderstanding Complex Systems.}
\newblock Springer, 2013.

\bibitem{Moder}
Joseph~J. Moder and Cecil~R. Phillips.
\newblock {\em Project management with CPM and PERT}.
\newblock Reinhold industrial engineering and management sciences textbook
  series. Reinhold Pub. Corp., 2 edition, 1970.

\bibitem{Mohri:2002:SFA:639508.639512}
Mehryar Mohri.
\newblock Semiring frameworks and algorithms for shortest-distance problems.
\newblock {\em J. Autom. Lang. Comb.}, 7(3):321--350, 2002.

\bibitem{journals/corr/abs-1106-2134}
Vincenzo Nicosia, John Tang, Mirco Musolesi, Giovanni Russo, Cecilia Mascolo,
  and Vito Latora.
\newblock Components in time-varying graphs.
\newblock {\em CoRR}, abs/1106.2134, 2011.

\bibitem{PraprotnikAMC}
Selena Praprotnik and Vladimir Batagelj.
\newblock Spectral centrality measures in temporal networks.
\newblock {\em Submitted to Ars Mathematica Contemporanea}, 2015.

\bibitem{comb}
John Riordan.
\newblock {\em Introduction to Combinatorial Analysis}.
\newblock Dover Books on Mathematics. Wiley New York, 1958.

\bibitem{slice}
Nicola Santoro, Walter Quattrociocchi, Paola Flocchini, and Arnaud Casteigts.
\newblock Time-varying graphs and social network analysis: Temporal indicators
  and metrics.
\newblock {\em 3rd AISB Social Networks and Multiagent Systems Symposium
  (SNAMAS}, pages 32--38, 2011.

\bibitem{wasserman1994social}
Stanley Wasserman and Katherine Faust.
\newblock {\em Social network analysis: Methods and applications}.
\newblock Cambridge University Press, 1994.

\bibitem{Xuan_JoFoCS2002}
Bui~B. Xuan, Afonso Ferreira, and Aubin Jarry.
\newblock {Computing shortest, fastest, and foremost journeys in dynamic
  networks}.
\newblock {\em International Journal of Foundations of Computer Science},
  14(2):267--285, 2003.

\bibitem{zimmerman}
U.~Zimmerman.
\newblock {\em Annals of Discrete Mathematics: Linear and Combinatorial
  Optimization in Ordered Algebraic Structures}.
\newblock North-Holland, 1981.

\end{thebibliography}

\end{document}